\documentclass[submission,nomarks,colorlinks=true]{dmtcs-episciences}

\usepackage[utf8]{inputenc}
\usepackage{subfigure}

\usepackage{amsmath,amssymb,amsthm}
\usepackage{algorithm}
\usepackage[noend]{algpseudocode}
\usepackage{enumerate}
\usepackage{comment}
\usepackage{tikz}
\usepackage{gensymb}
\usepackage{mathtools}

\newtheorem{example}{Example}
\newtheorem{definition}{Definition}
\newtheorem{myrule}{Rule}
\newtheorem{theorem}{Theorem}
\newtheorem{lemma}{Lemma}
\newtheorem{proposition}{Proposition}
\newtheorem{corollary}{Corollary}

\newcommand{\cp}{\mathrel{\!\vbox{\offinterlineskip\ialign{%
    \hfil##\hfil\cr
    $\scriptstyle\circ$\cr
    \noalign{\kern-1.2ex}
    $\times$\cr
}}\!}}
\newcommand{\tcl}[1]{{#1^*}}
\newcommand{\ucl}[1]{{\mathcal{U}\!\left({#1}\right)}}
\newcommand{\ucmp}[1]{\overline{#1}^u}
\newcommand{\ucc}[1]{\ucmp{\tcl{#1}}}
\newcommand{\forced}{\mathop{~\Gamma~}}
\newcommand{\fcl}{\mathop{~\Gamma^*~}}

\tikzstyle{vertex}=[circle,fill=black!10,minimum size=20pt,inner sep=0pt]
\tikzstyle{edge} = [draw,thick,->]
\tikzstyle{tedge} = [draw,dashed,->]
\tikzstyle{uedge} = [draw,thick]
\tikzset{top/.style={baseline=(current bounding box.north)}}
\tikzset{mid/.style={baseline=(current bounding box.center)}}
\tikzset{gscale/.style={xscale=1.4,yscale=1.5}}
\newcommand{\tikzcaption}[1]{\node[below=6mm] at (current bounding box.base) {#1}}

\author{Henning Koehler}
\title{Modular decomposition of transitive graphs and transitively orienting their complements}
\affiliation{Massey University, New Zealand}
\received{2017-10-11}
\revised{-}
\accepted{-}

\begin{document}
\publicationdetails{VOL}{2017}{ISS}{NUM}{SUBM}
\maketitle

\begin{abstract}
~\\
The modular decomposition of a graph is a canonical representation of its modules.
Algorithms for computing the modular decomposition of directed and undirected graphs differ significantly, with the undirected case being simpler, and algorithms for directed graphs often work by reducing the problem to decomposing undirected graphs.
In this paper we show that transitive acyclic digraphs have the same strong modules as their undirected versions.
This simplifies reduction for transitive digraphs, requiring only the computation of strongly connected components.

Furthermore, we are interested in permutation graphs, where both the graph and its complement are transitively orientable.
Such graphs may be represented indirectly, as the transitive closure of a given graph.
For non-transitive graphs we present a linear-time algorithm which allows us to identify prime-free modules w.r.t their transitive closure, which speeds up both modular decomposition and transitive orientation for sparse graphs.

Finally, we show that any transitive orientation of a digraph's complement also transitively orients the complement of the digraph's transitive closure, allowing us to find such orientations in (near-)linear time.
\end{abstract}

\section{Introduction}

A \emph{module} $M$ of a graph $G=(V,E)$, which may be directed or undirected, is a non-empty subset of $V$ such that all vertices in $M$ have the same predecessors and successors in $V\setminus M$.
Decomposing a graph into modules can help with various graph problems, as it provides a way to reduce them to the same or similar (sometimes easier) problems on smaller graphs \cite{DBLP:journals/dm/McConnellS99}.
While graphs can possess exponentially many modules, all modules can be represented efficiently in linear space, and this representation, known as \emph{modular decomposition}, can be found in linear time \cite{DBLP:journals/dm/McConnellS99,DBLP:journals/dam/McConnellM05}.
As the algorithms needed to achieve linear running time are highly complex, later work has focused on achieving near-linear running times with simpler algorithms, for both undirected \cite{DBLP:journals/dmtcs/McConnellS00} and directed graphs \cite{DBLP:conf/swat/HabibMP04,DBLP:conf/icalp/TedderCHP08}.

As pointed out in \cite{DBLP:journals/dmtcs/McConnellS00}, the issue of \emph{transitive orientation} is closely related.
Here we seek to orient the edges of an undirected graph in such a way that the resulting digraph is transitive.
Not every graph can be oriented this way, and transitively orientable graphs are known as comparability graphs.
Comparability graphs whose complement is also transitively orientable are called permutation graphs.

We are particularly interested in digraphs whose complement can be transitively oriented.
It turns out that the transitive closure of such a graph is a permutation graph, and that a pair of total orders characterizing the transitive closure can be computed in near-linear time.

To begin, we introduce some helpful tools and terminology.
Throughout the paper we consider all graphs to be simple and directed, and represent undirected edges as pairs of directed edges.

\pagebreak[3]
\begin{definition}[inverse, undirected, complement, oriented, transitive]~\\[-1.5em]
\begin{enumerate}[(i)]
\setlength\itemsep{-2pt}
\item The \emph{inverse} of an edge $(a,b)$ is the edge $(a,b)^{-1}=(b,a)$.
The inverse of a graph $G=(V,E)$ is the graph $G^{-1}=(V,E^{-1})$, where $E^{-1}=\{(a,b)^{-1}\mid (a,b)\in E\}$.
\item We say that a graph $G=(V,E)$ is \emph{undirected} iff $E=E^{-1}$.
The \emph{undirected closure} of a graph $G=(V,E)$ is the graph $\mathcal{U}(G)=(V,E\cup E^{-1})$.
\item The \emph{complement} of $G=(V,E)$ is the graph $\overline{G}=(V,V\cp V\setminus E)$, where $V\cp V=\{(a,b)\mid a,b\in V, a\neq b\}$.
The \emph{undirected complement} $\ucmp{G}$ of $G$ is the complement of $\mathcal{U}(G)$.
\item A graph $G=(V,E)$ is \emph{oriented} iff $E\cap E^{-1}=\emptyset$.
A graph $G'=(V,E')$ is an \emph{orientation} of $G=(V,E)$ iff $E'$ is a maximal oriented subset of $E$.
\item A graph $G=(V,E)$ is \emph{transitive} iff for all edges $(a,b),(b,c)\in E$ with $a\neq c$ we also have $(a,c)\in E$.
The \emph{transitive closure} of a graph $G=(V,E)$ is $\tcl{G}=(V,\tcl{E})$, where $\tcl{E}$ is the minimal transitive superset of $E$.
\end{enumerate}
\end{definition}

We shall describe the representation of modules next.

\subsection{Modular Decomposition}

The \emph{trivial modules} of $G$ are $V$ and all singleton sets.
A graph is called \emph{prime} if it does not possess non-trivial modules.
Two modules \emph{overlap} if they intersect, but neither one contains the other.
A module is a \emph{strong module} if it does not overlap with any module, otherwise it is a \emph{weak module}.
Given a partition $\mathcal{P}$ of $V$ into modules, also referred to as a \emph{congruence partition}, the \emph{quotient} $G/\mathcal{P}$ is the graph obtained from $G$ by merging the nodes of each module in $\mathcal{P}$.

\begin{example}\label{E:modules}
Consider the transitive closure $\tcl{G}$ of the graph $G$ shown on the left below.
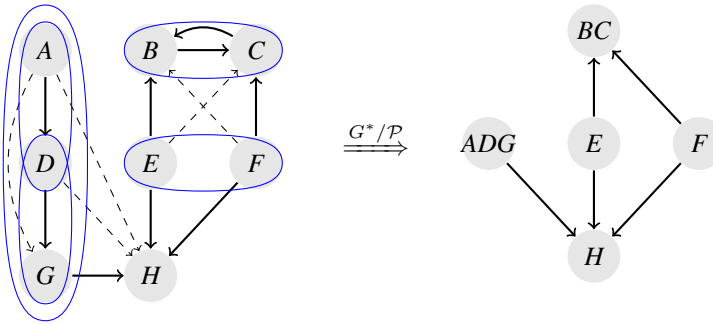
\begin{center}
\begin{tikzpicture}[gscale,top]
\foreach \pos/\name in {%
{(0,2)/A}, {(1,2)/B}, {(2,2)/C},
{(0,1)/D}, {(1,1)/E}, {(2,1)/F},
{(0,0)/G}, {(1,0)/H}}
    \node[vertex] (\name) at \pos {\name};
\foreach \source/\dest in {%
A/D, B/C, D/G, E/B, E/H, F/C, F/H, G/H}
    \path[edge] (\source) -- (\dest);
\path[edge] (C) to [out=150,in=30] (B);
\foreach \source/\dest in {%
A/H, D/H, E/C, F/B}
    \path[tedge] (\source) -- (\dest);
\path[tedge] (A) to [bend right=30] (G);
\draw [blue] plot [smooth cycle, tension=1.25] coordinates {(.75,2) (1.5,1.75) (2.25,2) (1.5,2.25)}; 
\draw [blue] plot [smooth cycle, tension=1.25] coordinates {(.75,1) (1.5,.75) (2.25,1) (1.5,1.25)}; 
\draw [blue] plot [smooth cycle, tension=1.25] coordinates {(0,.75) (.25,1.5) (0,2.25) (-.25,1.5)}; 
\draw [blue] plot [smooth cycle, tension=1.25] coordinates {(0,-.25) (.25,.5) (0,1.25) (-.25,0.5)}; 
\draw [blue] plot [smooth cycle, tension=1.25] coordinates {(0,-.4) (.4,1) (0,2.4) (-.4,1)}; 
\end{tikzpicture}
\qquad\raisebox{-2cm}{$\xRightarrow{\tcl{G}/\mathcal{P}}$}\qquad
\begin{tikzpicture}[gscale,top]
\foreach \pos/\name in {%
{(1,2)/BC},
{(0,1)/ADG}, {(1,1)/E}, {(2,1)/F},
{(1,0)/H}}
    \node[vertex] (\name) at \pos {\makebox[0pt][c]{\name}};
\foreach \source/\dest in {%
ADG/H, E/BC, E/H, F/BC, F/H}
    \path[edge] (\source) -- (\dest);
\end{tikzpicture}
\end{center}
The non-trivial modules of $\tcl{G}$ are AD, DG, ADG, BC and EF.
Its non-trivial strong modules are ADG, BC and EF.
The set $\mathcal{P}=\{ADG, BC, E, F, H\}$ is one of many congruence partitions.
The quotient graph $\tcl{G}/\mathcal{P}$ is shown on the right.
Unlike its transitive closure, $G$ is prime.
\qed
\end{example}

The \emph{modular decomposition} of $G$ is an $O(n)$-space representation of all modules of $G$.
It was first described in \cite{gallai67}, for a recent survey see \cite{DBLP:journals/csr/HabibP10}.
The decomposition forms a tree, whose nodes are the strong modules of $G$, ordered by the subset relationship.
In particular, the root of the modular decomposition tree is $V$, and its leaves are the singleton sets.

Furthermore, we can associate the \emph{child quotient} graph $(G|N)/\text{children}(N)$ with each non-leaf node $N$, where $G|N$ denotes the subgraph of $G$ on $N\subseteq V$, and children($N$) the maximal strong modules properly contained in $N$.
Each non-leaf node is classified based on the type of its child quotient graph:
\begin{itemize}
\setlength\itemsep{-5pt}
\item prime: the child quotient graph is a prime graph
\item series: the child quotient graph is a clique
\item parallel: the child quotient graph is independent (contains no edges)
\item ordered: the child quotient graphs describes a total order (directed graph only)
\end{itemize}
No other cases exist, and the types, together with an ordering of the vertices in the case of an ordered node, fully describe all modules of the child quotient graph, which in turn are modules of $G$.

\begin{example}
Consider again the graph from Example~\ref{E:modules}.
The modular decomposition tree of its transitive closure is shown on the right.
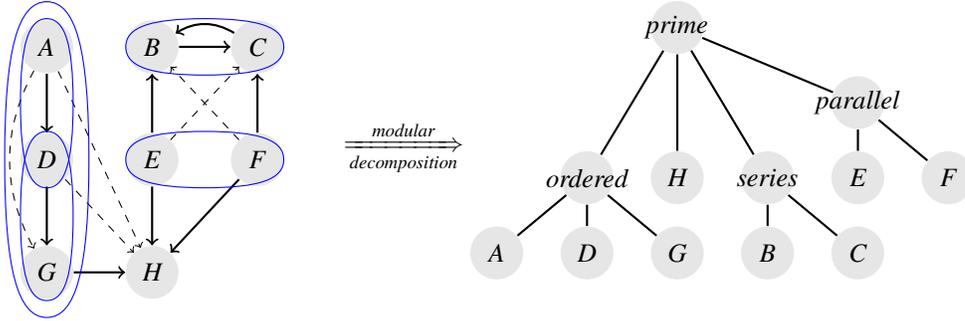
\vspace{-1em}
\begin{center}
\begin{tikzpicture}[gscale,top]
\foreach \pos/\name in {%
{(0,2)/A}, {(1,2)/B}, {(2,2)/C},
{(0,1)/D}, {(1,1)/E}, {(2,1)/F},
{(0,0)/G}, {(1,0)/H}}
    \node[vertex] (\name) at \pos {\name};
\foreach \source/\dest in {%
A/D, B/C, D/G, E/B, E/H, F/C, F/H, G/H}
    \path[edge] (\source) -- (\dest);
\path[edge] (C) to [out=150,in=30] (B);
\foreach \source/\dest in {%
A/H, D/H, E/C, F/B}
    \path[tedge] (\source) -- (\dest);
\path[tedge] (A) to [bend right=30] (G);
\draw [blue] plot [smooth cycle, tension=1.25] coordinates {(.75,2) (1.5,1.75) (2.25,2) (1.5,2.25)}; 
\draw [blue] plot [smooth cycle, tension=1.25] coordinates {(.75,1) (1.5,.75) (2.25,1) (1.5,1.25)}; 
\draw [blue] plot [smooth cycle, tension=1.25] coordinates {(0,.75) (.25,1.5) (0,2.25) (-.25,1.5)}; 
\draw [blue] plot [smooth cycle, tension=1.25] coordinates {(0,-.25) (.25,.5) (0,1.25) (-.25,0.5)}; 
\draw [blue] plot [smooth cycle, tension=1.25] coordinates {(0,-.4) (.4,1) (0,2.4) (-.4,1)}; 
\end{tikzpicture}
\qquad\raisebox{-2cm}{$\xRightarrow[\text{decomposition}]{\text{modular}}$}
\begin{tikzpicture}[xscale=1.2,top]
\foreach \pos/\name/\label in {%
{(2,3)/X/prime},
{(1,1)/ADG/ordered},
{(3,1)/BC/series},
{(4,2)/EF/parallel},
{(0,0)/A/A}, {(1,0)/D/D}, {(2,0)/G/G}, {(2,1)/H/H}, {(3,0)/B/B}, {(4,0)/C/C}, {(4,1)/E/E}, {(5,1)/F/F}}
    \node[vertex] (\name) at \pos {\makebox[0pt][c]{\label}};
\tikzstyle{edge} = [draw,thick]
\foreach \source/\dest in {%
X/ADG, X/H, X/BC, X/EF, ADG/A, ADG/D, ADG/G, BC/B, BC/C, EF/E, EF/F}
    \path[edge] (\source) -- (\dest);
\end{tikzpicture}
\end{center}
Together with a total ordering $A<D<G$ of the children of each ordered node, this describes all modules of $\tcl{G}$.
For parallel and series nodes all subsets of children are modules, and for ordered nodes all sub-intervals of children are, in our case AD and DG (but not AG).
\qed
\end{example}

The modular decomposition of a graph is a useful tool for many graph problems.
In particular, one may observe that the child quotient graphs of ordered, series and parallel nodes are permutation graphs, and that transitive orientations for them and their complements are trivial to find.

\subsection{Transitive Orientation}\label{S:TO}

We call a graph \emph{transitively orientable} iff it possesses a transitive orientation.
Transitive orientations of graphs are typically represented by providing a total ordering of vertices, which applies only to edges in the graph to be oriented.
As this representation only requires space (near)-linear in the number of vertices, it can be computed efficiently even if the graph is large - e.g. the complement of a given sparse graph.

Existence of a transitive orientation is particularly interesting for the complement of a transitive DAG $G$, as it characterizes $G$ as a permutation graph.
In this case $G$ has order dimension two (or less), i.e., can be characterized as the intersection of two linear orders \cite{pos41}.
The linear orders describing $G$ can then be found by merging $G$ with the transitive orientation $\mathcal{O}$ of its complement, and with its inverse $\mathcal{O}^{-1}$.

\begin{example}\label{E:transitive-orientation}
Consider the graph shown on the left.
A transitive orientation of the undirected complement of its transitive closure is given in the center.
Merging the two graphs results in two linear orderings of vertices, the intersection of which characterizes the original graph.
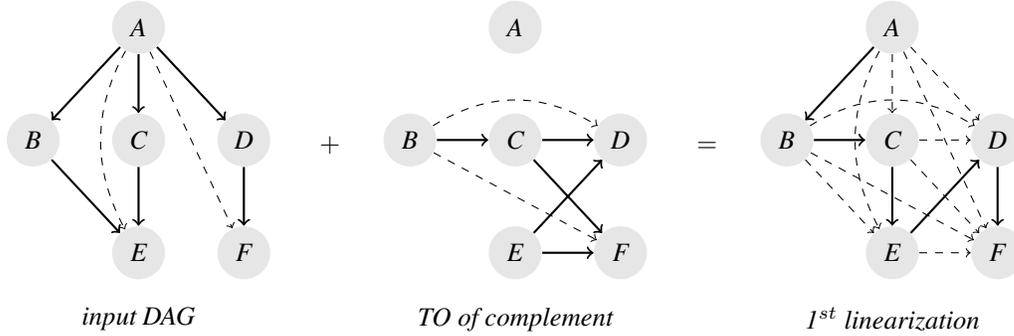
\newcommand{\vdef}{%
\foreach \pos/\name in {%
{(1,2)/A},
{(0,1)/B}, {(1,1)/C}, {(2,1)/D},
{(1,0)/E}, {(2,0)/F}}
    \node[vertex] (\name) at \pos {\name}}
\begin{center}
\hfill
\begin{tikzpicture}[gscale,top]
\vdef;
\foreach \source/\dest in {%
A/B, A/C, A/D, B/E, C/E, D/F}
    \path[edge] (\source) -- (\dest);
\tikzstyle{edge} = [draw,dashed,->]
\path[edge] (A) to [out=-120,in=120] (E);
\path[edge] (A) -- (F);
\tikzcaption{input DAG};
\end{tikzpicture}
\hfill\raisebox{-2cm}{$+$}\hfill
\begin{tikzpicture}[gscale,top]
\vdef;
\foreach \source/\dest in {%
B/C, C/D, C/F, E/D, E/F}
    \path[edge] (\source) -- (\dest);
\path[tedge] (B) -- (F);
\path[tedge] (B) to [bend left=30] (D);
\tikzcaption{TO of complement};
\end{tikzpicture}
\hfill\raisebox{-2cm}{$=$}\hfill
\begin{tikzpicture}[gscale,top]
\vdef;
\foreach \source/\dest in {%
A/B, B/C, C/E, E/D, D/F}
    \path[edge] (\source) -- (\dest);
\foreach \source/\dest in {%
A/C, A/D, A/F, B/E, B/F, C/D, C/F, E/F}
    \path[tedge] (\source) -- (\dest);
\path[tedge] (A) to [bend right=30] (E);
\path[tedge] (B) to [bend left=30] (D);
\tikzcaption{1$^{st}$ linearization};
\end{tikzpicture}
\hfill~
\end{center}
Merging $G$ with the inverse orientation of its complement provides the second linearization.
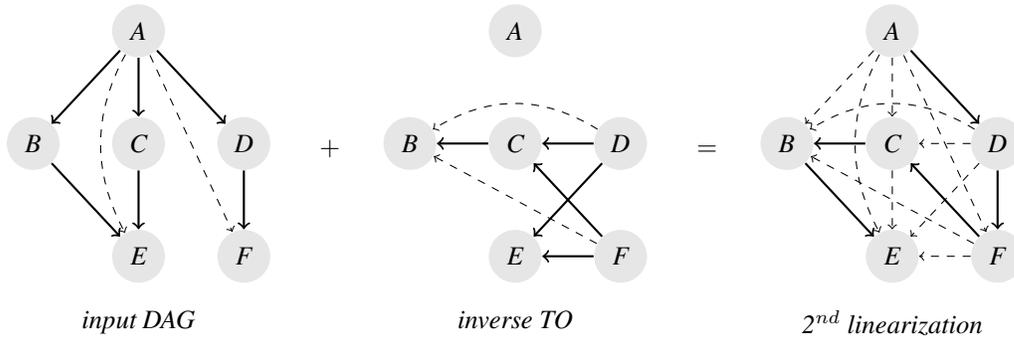
\begin{center}
\hfill
\begin{tikzpicture}[gscale,top]
\vdef;
\foreach \source/\dest in {%
A/B, A/C, A/D, B/E, C/E, D/F}
    \path[edge] (\source) -- (\dest);
\tikzstyle{edge} = [draw,dashed,->]
\path[edge] (A) to [out=-120,in=120] (E);
\path[edge] (A) -- (F);
\tikzcaption{input DAG};
\end{tikzpicture}
\hfill\raisebox{-2cm}{$+$}\hfill
\begin{tikzpicture}[gscale,top]
\vdef;
\foreach \source/\dest in {%
C/B, D/C, F/C, D/E, F/E}
    \path[edge] (\source) -- (\dest);
\path[tedge] (F) -- (B);
\path[tedge] (D) to [bend right=30] (B);
\tikzcaption{inverse TO};
\end{tikzpicture}
\hfill\raisebox{-2cm}{$=$}\hfill
\begin{tikzpicture}[gscale,top]
\vdef;
\foreach \source/\dest in {%
A/D, D/F, F/C, C/B, B/E}
    \path[edge] (\source) -- (\dest);
\foreach \source/\dest in {%
A/C, A/B, A/F, D/C, D/E, F/B, F/E, C/E}
    \path[tedge] (\source) -- (\dest);
\path[tedge] (A) to [bend right=30] (E);
\path[tedge] (D) to [bend right=30] (B);
\tikzcaption{2$^{nd}$ linearization};
\end{tikzpicture}
\hfill~
\end{center}
The partial ordering induced by $G$ can now be represented as ${\leq_{ABCEDF}} ~\cap~ {\leq_{ADFCBE}}$.
\qed
\end{example}

An efficient algorithm for combining the partial order described by $G^*$ with a transitive orientation of its complement, given as a linear ordering, can e.g. be found in \cite[Proof of Lemma~2]{DBLP:journals/jacm/EvenPL72}.

\section{Modular Decomposition of Transitive Graphs}

Most algorithms for decomposing directed graphs do so by reducing the problem to finding the modular decomposition of an undirected graph \cite{DBLP:conf/swat/HabibMP04,DBLP:conf/icalp/TedderCHP08}, or by adapting specific algorithms for undirected graphs \cite{DBLP:journals/dmtcs/DahlhausGM02}.
This suggests that although modular decomposition of directed and undirected graphs can both be performed in (near-) linear time, decomposing undirected graphs is technically simpler.

An obvious candidate for an undirected graph to which we might reduce the problem of modular decomposition of a digraph is its undirected closure.
Clearly, modules of the former must also be modules of the latter.
However, as the following example shows, their modular decompositions may differ greatly if the graph in question is not transitive.

\pagebreak[3]
\begin{example}
\newcommand{\vdef}{
\foreach \pos/\name in {%
{(0,2)/A},
{(0,1)/B}, {(1,1)/C},
{(0,0)/D}, {(1,0)/E}}
    \node[vertex] (\name) at \pos {\name}
}
Consider the directed graph $G$ shown below.
Its only non-trivial module is DE, leading to the decomposition tree on the right.
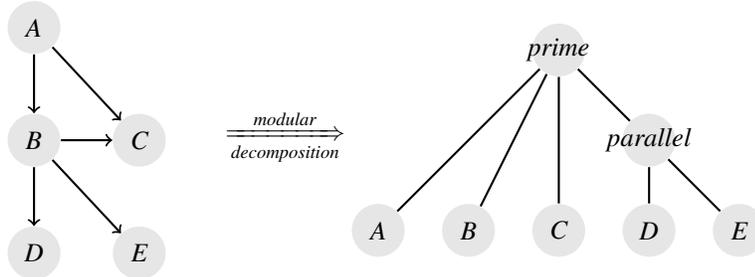
\begin{center}
\begin{tikzpicture}[gscale,mid]
\vdef;
\foreach \source/\dest in {%
A/B, A/C, B/C, B/D, B/E}
    \path[edge] (\source) -- (\dest);
\end{tikzpicture}
\qquad\raisebox{0cm}{$\xRightarrow[\text{decomposition}]{\text{modular}}$}
\begin{tikzpicture}[scale=1.2,mid]
\foreach \pos/\name/\label in {%
{(2,2)/X/prime},
{(3,1)/DE/parallel},
{(0,0)/A/A}, {(1,0)/B/B}, {(2,0)/C/C}, {(3,0)/D/D}, {(4,0)/E/E}}
    \node[vertex] (\name) at \pos {\makebox[0pt][c]{\label}};
\tikzstyle{edge} = [draw,thick]
\foreach \source/\dest in {%
X/A, X/B, X/C, X/DE, DE/D, DE/E}
    \path[uedge] (\source) -- (\dest);
\end{tikzpicture}
\end{center}
The undirected closure $\ucl{G}$ is shown below.
In addition to DE, the sets AC, ACD, ACE and ACDE are non-trivial modules of $\ucl{G}$.
Among these, only AC and ACDE are strong modules.
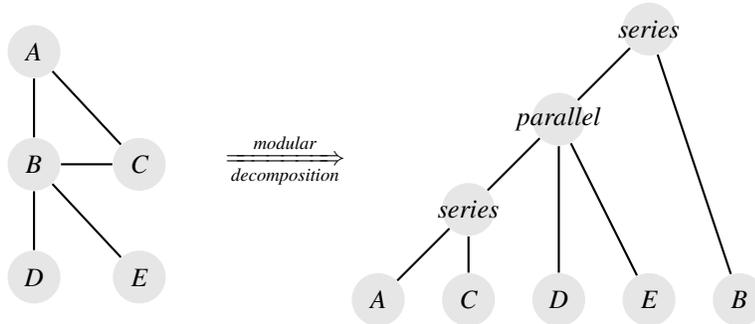
\begin{center}
\begin{tikzpicture}[gscale,mid]
\vdef;
\foreach \source/\dest in {%
A/B, A/C, B/C, B/D, B/E}
    \path[uedge] (\source) -- (\dest);
\end{tikzpicture}
\qquad\raisebox{0cm}{$\xRightarrow[\text{decomposition}]{\text{modular}}$}
\begin{tikzpicture}[scale=1.2,mid]
\foreach \pos/\name/\label in {%
{(3,3)/X/series},
{(2,2)/ACDE/parallel},
{(1,1)/AC/series},
{(0,0)/A/A}, {(1,0)/C/C}, {(2,0)/D/D}, {(3,0)/E/E}, {(4,0)/B/B}}
    \node[vertex] (\name) at \pos {\makebox[0pt][c]{\label}};
\tikzstyle{edge} = [draw,thick]
\foreach \source/\dest in {%
X/ACDE, X/B, ACDE/AC, ACDE/D, ACDE/E, AC/A, AC/C}
    \path[edge] (\source) -- (\dest);
\end{tikzpicture}
\end{center}
This shows that although modules of $G$ must also be modules of $\ucl{G}$, the two graphs may not share any non-trivial strong modules, leading to very different decomposition trees.
\qed
\end{example}

However, we seek to decompose \emph{transitive} graphs.
Here the relationship between the modular decompositions of $G$ and $\ucl{G}$ is much closer -- as we will see later (proof of Theorem~\ref{Th:transitive-strong-modules}, \emph{if} direction), strong modules of $\ucl{G}$ are strong modules of $G$.
However, $G$ may have strong modules which are only weak modules of $\ucl{G}$, and thus missing in the modular decomposition of $\ucl{G}$.

\begin{example}\label{Ex:transitive-cyclic}
\newcommand{\vdef}{
\foreach \pos/\name in {%
{(0.5,1)/A},
{(0,0)/B}, {(1,0)/C}}
    \node[vertex] (\name) at \pos {\name}
}
Consider the transitive graph $G$ and its undirected closure below:
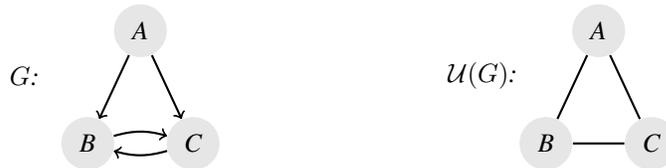
\begin{center}
\hfill
\begin{tikzpicture}[gscale,mid]
\vdef;
\foreach \source/\dest in {%
A/B, A/C}
    \path[edge] (\source) -- (\dest);
\path[edge] (B) to [bend left=15] (C);
\path[edge] (C) to [bend left=15] (B);
\node at (-0.6,0.6) {$G$:};
\end{tikzpicture}
\hfill
\begin{tikzpicture}[gscale,mid]
\vdef;
\foreach \source/\dest in {%
A/B, A/C, B/C}
    \path[uedge] (\source) -- (\dest);
\node at (-0.6,0.6) {$\ucl{G}$:};
\end{tikzpicture}
\hfill~
\end{center}
~\\
The set $\{B,C\}$ is a strong module of $G$, but only a weak module of $\ucl{G}$.\qed
\end{example}

The problem in Example~\ref{Ex:transitive-cyclic} hails from the fact that ordered and series nodes become indistinguishable in the undirected closure.
We can avoid this by identifying and merging strongly connected components.
This leaves us with a transitive DAG, whose modular decomposition is free of series nodes.
For such graphs, the modular decompositions of $G$ and its undirected closure $\ucl{G}$ become isomorphic.
Before we show this, we briefly establish some terminology and properties.

\begin{definition}[child quotient graph]~\\[-1.5em]
\label{D:module-type}
\begin{enumerate}[(i)]
\setlength\itemsep{-2pt}
\item We denote the minimal strong module containing a module $M$ by $M^+$.
\item The \emph{child quotient graph} of a weak module $M$ is $(G|M)/\mathcal{P}_M$, where $\mathcal{P}_M$ is the partition of $M$ into the maximal strong modules contained in $M$.
\item We call a weak module $M$ \emph{series}, \emph{parallel} or \emph{ordered} if its child quotient graph is a clique, independent or totally ordered, respectively.
\end{enumerate}
\end{definition}

One may easily confirm the following:

\begin{proposition}\label{P:parent-module}
A weak module $M$ is series/parallel/ordered iff $M^+$ is series/parallel/ordered.
\end{proposition}

\begin{proposition}\label{P:series-module}
If a module $M$ of $G=(V,E)$ can be partitioned into two non-empty subsets $M_1,M_2$ such that $M_1\times M_2\subseteq E$ and $M_2\times M_1\subseteq E$, then $M$ is series.
\end{proposition}

We are now ready to show our result.

\begin{theorem}\label{Th:transitive-strong-modules}
Let $G=(V,E)$ be a transitive acyclic digraph.
Then $M\subseteq V$ is a strong module of $G$ if and only if $M$ is a strong module of $\ucl{G}$.
\end{theorem}

\begin{proof}
(if) Let $M$ be a strong module of $\ucl{G}$, and assume to the contrary that $M$ is not even a module of $G$.
Then there must exist a \emph{splitter} vertex $x\in V\setminus M$ for which the predecessor/successor relationships to vertices in $M$ differ.
As $M$ is a module in $\ucl{G}$, $x$ must be adjacent to all vertices in $M$.
We can thus partition $M$ into non-empty sets $M_P=M\cap P(x)$ and $M_S=M\cap S(x)$, where $P(x)$ and $S(x)$ denote the predecessors and successors of $x$, respectively.
Since $G$ is transitive, there must exist an edge from every vertex in $M_P$ to every vertex in $M_S$.
By Proposition~\ref{P:series-module}, $M$ must be a series module in $\ucl{G}$.

Denote by $X$ the set of all vertices $x'\in V\setminus M$ with $M\cap P(x') = M_P$ and $M\cap S(x') = M_S$.
In particular we have $x\in X$, so $X$ is not empty.
We claim that $M\cup X$ is a series module in $\ucl{G}$.
To show this claim, let $y\in V\setminus(M\cup X)$.
The relationships between $y$, $X$, $M_P$ and $M_S$ are visualized below.
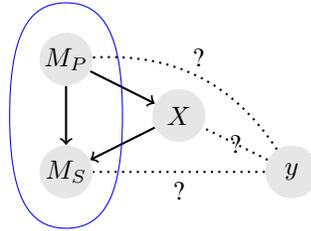
\begin{center}
\begin{tikzpicture}[scale=1.5,swap]
\foreach \pos/\name/\label in {%
{(0,1)/P/$M_P$},
{(1,0.5)/X/$X$},
{(0,0)/S/$M_S$}, {(2,0)/Y/$y$}}
    \node[vertex] (\name) at \pos {\label};
\foreach \source/\dest in {P/X, P/S, X/S}
    \path[edge] (\source) -- (\dest);
\path[uedge,dotted] (Y) edge[below] node{?} (S);
\path[uedge,dotted] (Y) edge node{?} (X);
\path[uedge,dotted] (Y) edge[bend right=30,above] node{?} (P);
\draw [blue] plot [smooth cycle, tension=1.25] coordinates {(0,-.5) (.5,.5) (0,1.5) (-.5,0.5)}; 
\end{tikzpicture}
\end{center}

If $y$ is not adjacent to any vertex in $M\cup X$, then $y$ does not split $M\cup X$.
Otherwise, if $y$ is adjacent to $x'\in X$, then either $(y,x')\in E$ or $(x',y)\in E$.
As $G$ is transitive, it follows that either $(y,s)\in E$ for all $s\in M_s$, or $(p,y)\in E$ for all $p\in M_p$.
In both cases $y$ is adjacent to a vertex in $M$.

Since $M$ is a module in $\ucl{G}$, $y$ must be adjacent to all vertices in $M$ whenever it is adjacent to one.
Since $y\notin X$, we either have $(y,p)\in E$ for some $p\in M_p$, or $(s,y)\in E$ for some $s\in M_S$.
As $G$ is transitive, it follows that either $(y,x')\in E$ for all $x'\in X$, or $(x',y)\in E$ for all $x'\in X$.
In both cases, $y$ is adjacent to all vertices in $M\cup X$, and hence does not split $M\cup X$ in $\ucl{G}$.

This shows the claim that $M\cup X$ is a module in $\ucl{G}$.
Since every vertex in $M$ is adjacent to every vertex in $X$, $M\cup X$ must be of the series type by Proposition~\ref{P:series-module}.
However, that would mean that both $M$ and its parent node in the modular decomposition of $\ucl{G}$ are of the series type, which contradicts $M$ being a strong module in $\ucl{G}$.
This disproved our assumption that $M$ is not a module in $G$.

Finally, since every module in $G$ is a module in $\ucl{G}$, every weak (i.e., overlapping) module of $G$ is a weak module of $\ucl{G}$.
Thus $M$ must be a strong module of $G$.

(only if) If $M$ is a strong module in $G$, then $M$ is clearly a module in $\ucl{G}$, and we only need to show that it is strong.
Denote by $M^+$ the minimal strong module in $\ucl{G}$ containing $M$.
If $M^+$ is prime, then $M=M^+$ follows immediately.
If $M^+$ is parallel in $\ucl{G}$ then it is parallel in $G$, and so is $M$ by Proposition~\ref{P:parent-module}, and $M=M^+$ follows again.
Finally, if $M^+$ is series in $\ucl{G}$, then the child quotient graphs of $M^+$ and $M$ in $G$ are oriented cliques.
Since $G$ is acyclic, this orientation must induce a total order, so $M^+$ and $M$ are both ordered modules in $G$, and $M=M^+$ follows once more.
\end{proof}

This immediately allows us to derive the modular decomposition of $G$ from that of $\ucl{G}$.

\begin{corollary}
Let $G=(V,E)$ be a transitive DAG.
Then the modular decomposition tree of $G$ can be obtained from the modular decomposition tree of $\ucl{G}$ by relabelling series nodes as ordered nodes.
\end{corollary}

\begin{proof}
By Theorem~\ref{Th:transitive-strong-modules} the modular decomposition trees of $G$ and $\ucl{G}$ can only differ in their labelling.
Prime and parallel nodes in $\ucl{G}$ are clearly prime and parallel nodes in $G$ as well.
As $G$ is acyclic, series nodes in $\ucl{G}$ must be ordered nodes in $G$.
\end{proof}

\section{Modular Decomposition w.r.t. Transitive Closure}\label{S:reduction}

\newcommand{\seqflow}{\texttt{seq}}
\newcommand{\parflow}{\texttt{par}}
\newcommand{\anc}{\texttt{anc}}
\newcommand{\desc}{\texttt{desc}}

Often the transitive graph we seek to decompose will not be given directly.
Instead, the input is an arbitrary directed graph $G$, and we wish to find a modular decomposition of its transitive closure $\tcl{G}$.

This could be approached by first computing $\tcl{G}$, and then applying some modular decomposition algorithm.
However, $\tcl{G}$ can be much larger that $G$, especially if $G$ is sparse.
We therefore seek to pre-process $G$ by identifying certain modules of $\tcl{G}$, which can then be used to reduce $G$ before computing its transitive closure.
While this does not improve the worst-case complexity for arbitrary graphs, it can greatly improve performance for sparse graphs, and is therefore highly relevant in practice.

We first identify the strongly connected components of $G$ and contract them, in linear time.
These are precisely the series nodes of $\tcl{G}$, which we insert into the modular decomposition tree afterwards.

Our main preprocessing algorithm then consists of two reduction rules, each of which merges two vertices into one.
They are applied in any order until no more rule applications are possible -- as we will show later, the resulting reduced graph is unique.

\begin{myrule}[sequential flow]~\\
If $b$ is the only successor of $a$, and $a$ the only predecessor of $b$, merge $a$ and $b$.
\end{myrule}

\begin{myrule}[parallel flow]~\\
If $a,b$ with $a\neq b$ have the same sets of predecessors and successors, merge $a$ and $b$.
\end{myrule}

One may easily observe that nodes merged using sequential flow form ordered modules w.r.t. $\tcl{G}$, while nodes merged using parallel flow are parallel modules w.r.t. $\tcl{G}$.
As part of the reduction process, we may therefore construct an \emph{unreduced} modular decomposition tree for each merged node-set.
Unreduced modular decompositions may contain superfluous nodes for weak modules, but these can easily be eliminated by deleting ordered and parallel nodes with parents of the same type \cite{DBLP:journals/dmtcs/McConnellS00}.

Applying these reduction rules as far as possible can be done in linear time, by keeping track of all possible rule application, with updates after each merge.
Parallel nodes can be identified efficiently by maintaining a hash table, with the pairs of sets of predecessor and successor vertices as keys.
To ensure linear time complexity regardless of node degree, we identify predecessor / successor sets by a fixed-size hash value that can be updated in constant time when vertices are merged by either rule.
One possible implementation is a simple XOR of cryptographic hash-values of vertices.

\begin{example}
Consider once more the graph from Example~\ref{E:transitive-orientation}.
Initially, we can apply two reduction rules, $\seqflow(D,F)$ and $\parflow(B,C)$, in either order.
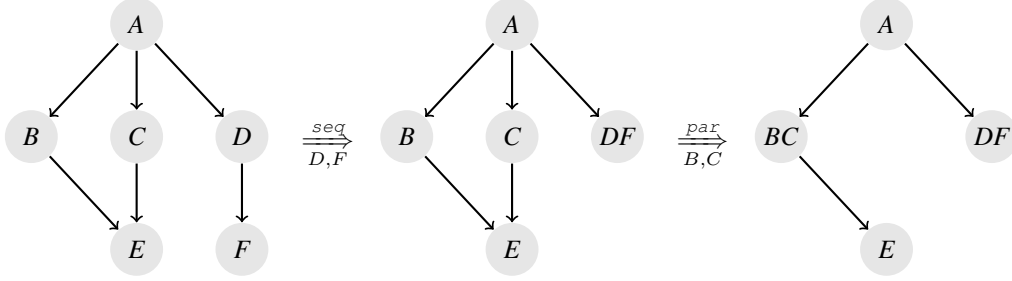
\begin{center}
\begin{tikzpicture}[gscale,top]
\foreach \pos/\name in {%
{(1,2)/A},
{(0,1)/B}, {(1,1)/C}, {(2,1)/D},
{(1,0)/E}, {(2,0)/F}}
    \node[vertex] (\name) at \pos {\name};
\foreach \source/\dest in {%
A/B, A/C, A/D, B/E, C/E, D/F}
    \path[edge] (\source) -- (\dest);
\end{tikzpicture}
\quad\raisebox{-2cm}{$\xRightarrow[D,F]{\seqflow}$}\quad
\begin{tikzpicture}[gscale,top]
\foreach \pos/\name in {%
{(1,2)/A},
{(0,1)/B}, {(1,1)/C}, {(2,1)/DF},
{(1,0)/E}}
    \node[vertex] (\name) at \pos {\name};
\foreach \source/\dest in {%
A/B, A/C, A/DF, B/E, C/E}
    \path[edge] (\source) -- (\dest);
\end{tikzpicture}
\quad\raisebox{-2cm}{$\xRightarrow[B,C]{\parflow}$}\quad
\begin{tikzpicture}[gscale,top]
\foreach \pos/\name in {%
{(1,2)/A},
{(0,1)/BC}, {(2,1)/DF},
{(1,0)/E}}
    \node[vertex] (\name) at \pos {\name};
\foreach \source/\dest in {%
A/BC, A/DF, BC/E}
    \path[edge] (\source) -- (\dest);
\end{tikzpicture}
\end{center}
After a reduction step, new rule applications involving the merged vertex may become available:
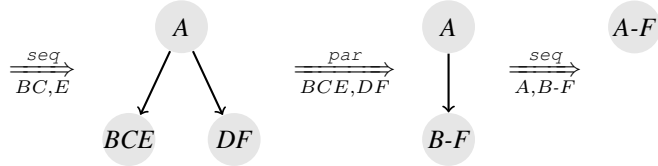
\begin{center}
\quad\raisebox{-1cm}{$\xRightarrow[BC,E]{\seqflow}$}\quad
\begin{tikzpicture}[gscale,top]
\foreach \pos/\name in {%
{(0.5,2)/A},
{(0,1)/BCE}, {(1,1)/DF}}
    \node[vertex] (\name) at \pos {\name};
\foreach \source/\dest in {%
A/BCE, A/DF}
    \path[edge] (\source) -- (\dest);
\end{tikzpicture}
\quad\raisebox{-1cm}{$\xRightarrow[BCE,DF]{\parflow}$}\quad
\begin{tikzpicture}[gscale,top]
\foreach \pos/\name/\label in {%
{(0,2)/A/A},
{(0,1)/B/B-F}}
    \node[vertex] (\name) at \pos {\label};
\path[edge] (A) -- (B);
\end{tikzpicture}
\quad\raisebox{-1cm}{$\xRightarrow[A,B\text{-}F]{\seqflow}$}\quad
\begin{tikzpicture}[scale=1.5,auto,swap,top]
\node[vertex] (A) at (0,0) {A-F};
\end{tikzpicture}
\end{center}
~\\[5pt]
In this particular instance the graph reduces to a single vertex, making it trivial to decompose.
\qed
\end{example}

We observe that no vertex can meet the conditions of both a sequential and parallel flow rule application at the same time.
This simplifies processing as it restricts interaction between rule applications.
We write $\seqflow(a,b)$ ($\parflow(a,b)$) to indicate that the conditions of the sequential (parallel) flow rule are met for $a,b$.

\begin{lemma}\label{L:flow-rules}
Let $a,b$ be two vertices in $G=(V,E)$ with $\seqflow(a,b)$.
Then there does not exist any vertex $x\in V$ with $\parflow(a,x)$ or $\parflow(b,x)$.
\end{lemma}

\begin{proof}
As $a$ is the only predecessor of $b$, no other vertex $x$ has the same successors as $a$.
As $b$ is the only successor of $a$, no other vertex $x$ has the same predecessors as $b$.
\end{proof}

Next, we show that the order of rule application does not affect the final result.
This property of rule systems is typically referred to as \emph{confluence}.

\begin{theorem}
The sequential/parallel flow rules are confluent.
\end{theorem}

\begin{proof}
By Newman's diamond lemma \cite{diamond-lemma} it suffices to show that the rules are locally confluent, i.e., that any two graphs obtained from a graph $G$ via a single rule application can be reduced to the same graph by further rule applications.
This is trivial if the two rule applications do not interact, and by Lemma~\ref{L:flow-rules} interaction can only occur for pairs of parallel or pairs of sequential rules that share a vertex.
In both cases all three vertices involved in the pair of rule applications can be merged into a single vertex with another application of the same rule.
\end{proof}

Intuitively, our reduction steps identify non-prime modules in a bottom-up fashion, until they encounter a prime module or fail to identify an ordered or parallel node.
While there is little to be done about prime modules, we find that failure to identify ordered or parallel nodes is always due to transitive edges.

\begin{definition}[prime-free]~\\
We call a graph \emph{prime-free} iff it does not contain any prime modules.
We call a module $M$ of graph $G$ \emph{prime-free} iff $G|M$ is prime-free.
\end{definition}

\begin{theorem}
Let $G$ be acyclic and transitively reduced.
Then sequential flow and parallel flow rules eliminate all prime-free modules of $\tcl{G}$.
\end{theorem}

\begin{proof}
Since non-prime modules are made up of overlapping 2-sets of their children, all of which are modules, it suffices to show that all all modules of size 2 are merged.
Denote by $\anc(x)$, $\desc(x)$ the ancestors and descendants of vertex $x$ w.r.t. $G$, respectively.
Now let $\{a,b\}\subseteq V$ be a module w.r.t. $\tcl{G}$.
Then $a,b$ have the same ancestors and descendants outside of $\{a,b\}$, that is
\[
\anc(a)\setminus\{a,b\} = \anc(b)\setminus\{a,b\}
\qquad\text{and}\qquad
\desc(a)\setminus\{a,b\} = \desc(b)\setminus\{a,b\}
\]
As $G$ is acyclic, it cannot contain both arcs $a\to b$ and $b\to a$.
This leaves two possibilities.

If $G$ contains an arc between $a$ and $b$, say $a\to b$, then $\anc(b) = \{a\} \cup \anc(a)$ and $\desc(a) = \{b\} \cup \desc(b)$.
Since $G$ is transitively reduced, it follows that $a$ is the only parent of $b$, and $b$ the only child of $a$.
Thus $a,b$ will be merged by the sequential flow rule.

If $G$ contains no arc between $a$ and $b$, then $\anc(a) = \anc(b)$ and $\desc(a) = \desc(b)$.
Since $G$ is transitively reduced, it follows that $a$ and $b$ share the same parents and children in $G$, and thus will be merged using the parallel flow rule.
\end{proof}

In particular we find that transitively reduced prime-free DAGs (e.g. all trees) are reduced to a single vertex.
If $G$ is transitively reduced but its modular decomposition contains a single prime-node as the root, then $G$ reduces to the child-quotient graph of this prime node.
In either case, no further decomposition is needed, provided it is known a priori that the reduced graph is prime.

\section{Transitive Orientation of Complement Graphs}

We now seek to find a transitive orientation of the complement of a transitive graph $\tcl{G}$, provided one exists.
The challenge here is to do so in time (near-) linear in the size of $G$.
We won't be able to achieve this for arbitrary digraphs, but we can for graphs $G$ with transitively orientable complement.

The \emph{Ordered Vertex Partitioning} (OVP) algorithm in \cite{DBLP:journals/dmtcs/McConnellS00} provides both a modular decomposition, and an orientation of a given undirected graph, with the provision that the orientation is transitive if the graph is transitively orientable.
However, we do not wish to orient $\ucl{\tcl{G}}$, but its complement $\ucc{G}$, which may be significantly bigger when $G$ is sparse.

Luckily, it turns out that the OVP algorithm can easily be modified to operate on the complement of the given graph, by switching the role of neighbours and non-neighbours in the central \emph{split} operation.
As already observed in \cite[Section~11.2]{DBLP:journals/dmtcs/DahlhausGM02}, this does not affect the complexity of the algorithm, allowing us to compute a transitive orientation of $\ucc{G}$ in time near-linear in the size of $\tcl{G}$.

The remaining problem is that $\tcl{G}$ can be much larger than $G$, and while the reduction steps of Section~\ref{S:reduction} can be applied here as well (transitively orienting the complements of ordered or parallel quotient graphs is trivial), any approach involving transitive closure computation won't be able to achieve near-linear complexity.
To maintain the complexity bound one could simply apply the modified OVP algorithm to find a transitive orientation of $\ucl{G}$ rather than $\ucc{G}$.
As we will show in the following, this approach delivers a transitive orientation for $\ucc{G}$ whenever $\ucl{G}$ is transitively orientable.

The basis for computing transitive orientations is the forcing relationship between edges, where orientation of one edge forces the orientation of another to maintain transitivity.

\begin{definition}[forced orientation]~\\
Let $G$ be an undirected graph.
We say that two edges $(a,b),(c,d)$ in $G$ \emph{force each other}, denoted as $(a,b)\forced(c,d)$, iff either $a=c$ and $b,d$ are not adjacent in $G_c$, or $b=d$ and $a,c$ are not adjacent in $G$.
We denote the transitive closure of relation $\forced$ by $\fcl$.
\end{definition}

It has been shown in \cite[Proof of Lemma 2.1]{DBLP:conf/soda/McConnellS94} that if $G$ is prime, then orienting one edge forces the orientation of all other edges.
Orienting the quotient graphs of series and parallel nodes is trivial.

\begin{example}\label{E:forced}
Consider the graph $G$ below, which is the undirected complement of the transitive closure of the graph from Example~\ref{E:transitive-orientation}.
If we pick an orientation of one edge, e.g. $E\to F$, then the orientation of most other edges is forced by it.
E.g. $E\to F$ forces $E\to D$, which in turn forces $B\to D$.
\newcommand{\vdef}{%
\foreach \pos/\name in {%
{(1,2)/A},
{(0,1)/B}, {(1,1)/C}, {(2,1)/D},
{(1,0)/E}, {(2,0)/F}}
    \node[vertex] (\name) at \pos {\name}}
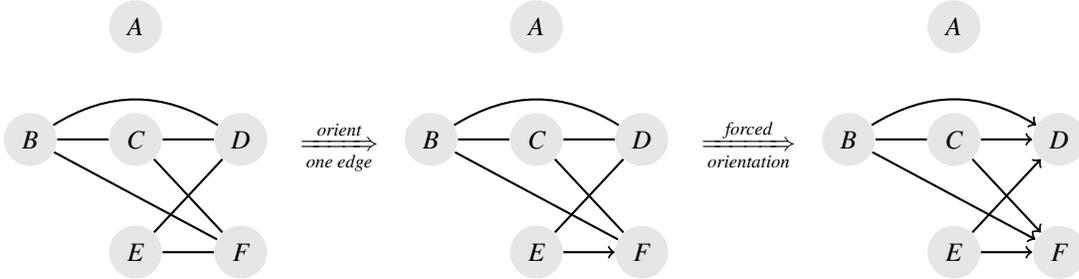
\begin{center}
\begin{tikzpicture}[gscale,top]
\vdef;
\foreach \source/\dest in {%
B/C, B/F, C/D, C/F, E/D, E/F}
    \path[uedge] (\source) -- (\dest);
\path[uedge] (B) to [bend left=30] (D);
\end{tikzpicture}
\quad\raisebox{-2cm}{$\xRightarrow[\text{one edge}]{\text{orient}}$}\quad
\begin{tikzpicture}[gscale,top]
\vdef;
\foreach \source/\dest in {%
B/C, B/F, C/D, C/F, E/D}
    \path[uedge] (\source) -- (\dest);
\path[uedge] (B) to [bend left=30] (D);
\path[edge] (E) to (F);
\end{tikzpicture}
\quad\raisebox{-2cm}{$\xRightarrow[\text{orientation}]{\text{forced}}$}\quad
\begin{tikzpicture}[gscale,top]
\vdef;
\path[uedge] (B) to (C);
\foreach \source/\dest in {%
B/F, C/D, C/F, E/D, E/F}
    \path[edge] (\source) -- (\dest);
\path[edge] (B) to [bend left=30] (D);
\end{tikzpicture}
\end{center}
The only edge not forced this way is the one between B and C.
This happens because $G$ is not a prime graph, and BCE is a module of $G$.
Orienting the missing edge as $B\to C$ yields the transitive orientation from Example~\ref{E:transitive-orientation}.
Orienting it as $C\to B$ yields another transitive orientation.
\qed
\end{example}

A critical observation in relating transitive orientations of $\ucc{G}$ and $\ucmp{G}$ is that the transitive closures of their forcing relationships are closely related.

\begin{lemma}\label{L:forcing}
Let $(a,b),(c,d)$ be two edges in $\ucc{G}$ that indirectly force each other in $\ucc{G}$, i.e., we have $(a,b) \fcl (c,d)$.
Then $(a,b),(c,d)$ indirectly force each other in $\ucmp{G}$.
\end{lemma}

\begin{proof}
Let $(a,x),(b,x)$ be two edges in $\ucc{G}$ that force each other directly.
Then either $(a,b)$ or $(b,a)$ lies in $\tcl{G}$, say $(a,b)$.
This means $G$ contains a path $a\to v_1 \to \ldots \to v_n \to b$, for some $n$.

Assume $(x,v_i)$ does not lie in $\ucmp{G}$, for $1\leq i\leq n$.
Then either $(x,v_i)$ or $(v_i,x)$ must lie in $G$.
If $(x,v_i)$ lies in $G$, then $(x,b)$ lies in $\tcl{G}$, contradicting $(b,x)\in\ucc{G}$.
If $(v_i,x)$ lies in $G$, then $(a,x)$ lies in $\tcl{G}$, contradicting $(a,x)\in\ucc{G}$.

Hence $(x,v_i)$ must lie in $\ucmp{G}$ for all $i=1\ldots n$, so the following forcing relationships hold in $\ucmp{G}$:
\[
(a,x)
\quad\Gamma\quad (v_1,x)
\quad\Gamma\quad (v_2,x)
\quad\Gamma\quad \ldots
\quad\Gamma\quad (v_n,x)
\quad\Gamma\quad (b,x)
\]

This shows the claim for edges that directly force each other in $\ucc{G}$.
For edges that force each other indirectly, the claim then follows by transitivity of $\fcl$.
\end{proof}

As a consequence of Lemma~\ref{L:forcing}, any transitive orientation of $\ucmp{G}$ is a transitive orientation of $\ucc{G}$ as well, so for graphs $G$ where $\ucmp{G}$ is transitively orientable, we can find a transitive orientation of $\ucc{G}$ in time near-linear in the size of $G$.

\begin{theorem}\label{Th:to-tcc}
Any transitive orientation of $\ucmp{G}$ is a transitive orientation of $\ucc{G}$.
\end{theorem}

\begin{proof}
An orientation of an undirected graph is a transitive orientation iff all pairs of edges forcing each other are suitably oriented (i.e., according to the forcing relationship).
By Lemma~\ref{L:forcing} any two edges in $\ucc{G}$ forcing each other also force each other in $\ucmp{G}$, and thus are suitably oriented.
\end{proof}

In particular this gives us:

\begin{corollary}
Let $G$ be a digraph such that $\ucmp{G}$ is transitively orientable.
Then we can find a transitive orientation of $\ucc{G}$ in near-linear time.
\end{corollary}

\begin{proof}
As $\ucmp{G}$ is transitively orientable, we can find a transitive orientation of $\ucmp{G}$ in near-linear time using a modified OVP Algorithm.
By Theorem~\ref{Th:to-tcc} this gives us a transitive orientation of $\ucc{G}$.
\end{proof}

We must point out however that $\ucmp{G}$ may not be transitively orientable even though $\ucc{G}$ is.

\begin{example}\label{Ex:to-tcc}
\newcommand{\vdef}{%
\foreach \pos/\name in {%
{(0.5,1.5)/A},
{(-.3,0.95)/B}, {(1.3,0.95)/C},
{(0,0)/D}, {(1,0)/E}}
    \node[vertex] (\name) at \pos {$\name$}
}
Consider the graph $G$ shown below. Its complement $\ucmp{G}$ is shown on the right.
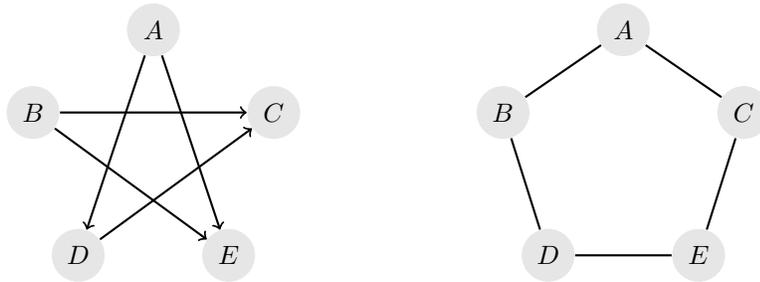
\begin{center}
\hfill
\begin{tikzpicture}[scale=2,swap]
\vdef;
\foreach \source/\dest in {A/D, A/E, B/C, B/E, D/C}
    \path[edge] (\source) -- (\dest);
\end{tikzpicture}
\hfill
\begin{tikzpicture}[scale=2,swap]
\vdef;
 complement edges
\foreach \source/\dest in {A/B, A/C, B/D, C/E, D/E}
    \path[uedge] (\source) -- (\dest);
\end{tikzpicture}
\hfill~
\end{center}
Observe that $\ucmp{G}$ contains an odd-length cycle of forced edges, and thus cannot be transitively oriented.
However, the transitive closure $\tcl{G}$ of $G$ does have a complement that is transitively orientable:
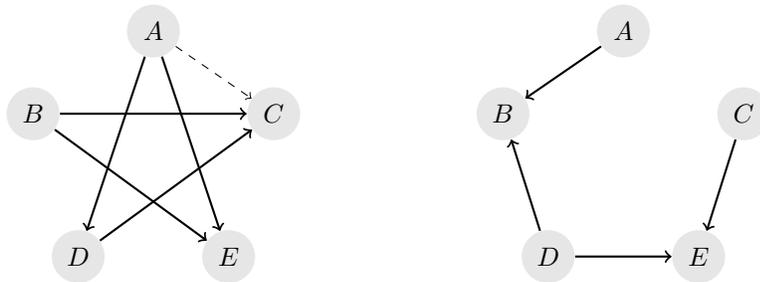
\begin{center}
\hfill
\begin{tikzpicture}[scale=2,swap]
\vdef;
\foreach \source/\dest in {A/D, A/E, B/C, B/E, D/C}
    \path[edge] (\source) -- (\dest);
\path[tedge] (A) -- (C);
\end{tikzpicture}
\hfill
\begin{tikzpicture}[scale=2,swap]
\vdef;
 complement edges
\foreach \source/\dest in {A/B, D/B, C/E, D/E}
    \path[edge] (\source) -- (\dest);
\end{tikzpicture}
\hfill~
\end{center}
\qed
\end{example}

\section{Conclusion}

We have presented several results related to finding the modular decomposition of a transitive graph $G$, and a transitive orientation of its complement, in particular when $G$ is given indirectly as the transitive closure of a sparse graph.
This is particularly interesting for digraphs with transitively orientable complement, for which we can identify, in near-linear time, a pair of total orders characterizing the transitive closure.
Among other applications, this allows answering of reachability queries in constant time.

\bibliography{graph.bib}{}

\begin{thebibliography}{10}

\bibitem{DBLP:journals/dmtcs/DahlhausGM02}
Elias Dahlhaus, Jens Gustedt, and Ross~M. McConnell.
\newblock Partially complemented representations of digraphs.
\newblock {\em Discrete Mathematics {\&} Theoretical Computer Science},
  5(1):147--168, 2002.

\bibitem{pos41}
Ben Dushnik and E.~W. Miller.
\newblock Partially ordered sets.
\newblock {\em American Journal of Mathematics}, 63(3):600--610, 1941.

\bibitem{DBLP:journals/jacm/EvenPL72}
Shimon Even, Amir Pnueli, and Abraham Lempel.
\newblock Permutation graphs and transitive graphs.
\newblock {\em J. {ACM}}, 19(3):400--410, 1972.

\bibitem{gallai67}
T.~Gallai.
\newblock Transitiv orientierbare graphen.
\newblock {\em Acta Mathematica Academiae Scientiarum Hungarica},
  18(1-2):25--66, 1967.

\bibitem{DBLP:conf/swat/HabibMP04}
Michel Habib, Fabien de~Montgolfier, and Christophe Paul.
\newblock A simple linear-time modular decomposition algorithm for graphs,
  using order extension.
\newblock In {\em {SWAT}}, pages 187--198, 2004.

\bibitem{DBLP:journals/csr/HabibP10}
Michel Habib and Christophe Paul.
\newblock A survey of the algorithmic aspects of modular decomposition.
\newblock {\em Computer Science Review}, 4(1):41--59, 2010.

\bibitem{DBLP:journals/dam/McConnellM05}
Ross~M. McConnell and Fabien de~Montgolfier.
\newblock Linear-time modular decomposition of directed graphs.
\newblock {\em Discrete Applied Mathematics}, 145(2):198--209, 2005.

\bibitem{DBLP:journals/dm/McConnellS99}
Ross~M. McConnell and Jeremy Spinrad.
\newblock Modular decomposition and transitive orientation.
\newblock {\em Discrete Mathematics}, 201(1-3):189--241, 1999.

\bibitem{DBLP:conf/soda/McConnellS94}
Ross~M. McConnell and Jeremy~P. Spinrad.
\newblock Linear-time modular decomposition and efficient transitive
  orientation of comparability graphs.
\newblock In {\em {SODA}}, pages 536--545, 1994.

\bibitem{DBLP:journals/dmtcs/McConnellS00}
Ross~M. McConnell and Jeremy~P. Spinrad.
\newblock Ordered vertex partitioning.
\newblock {\em Discrete Mathematics {\&} Theoretical Computer Science},
  4(1):45--60, 2000.

\bibitem{diamond-lemma}
M.~H.~A. Newman.
\newblock On theories with a combinatorial definition of ``equivalence''.
\newblock {\em Annals of Mathematics}, 43:223--243, 1942.

\bibitem{DBLP:conf/icalp/TedderCHP08}
Marc Tedder, Derek~G. Corneil, Michel Habib, and Christophe Paul.
\newblock Simpler linear-time modular decomposition via recursive factorizing
  permutations.
\newblock In {\em Automata, Languages and Programming ({ICALP})}, pages
  634--645, 2008.

\end{thebibliography}
\bibliographystyle{plain}

\end{document}